\documentclass[submission,copyright,creativecommons]{eptcs}
\providecommand{\event}{AiML 2026} 

\usepackage{aiml26}

\usepackage{iftex}

\ifpdf
  \usepackage{underscore}         
  \usepackage[T1]{fontenc}        
\else
  \usepackage{breakurl}           
\fi

\usepackage{graphicx}%
\usepackage{multirow}%
\usepackage{mathrsfs}%
\usepackage[title]{appendix}%
\usepackage{xcolor}%
\usepackage{textcomp}%
\usepackage{manyfoot}%
\usepackage{booktabs}%
\usepackage{marvosym }

\usepackage{url}
\usepackage{amsmath}
\usepackage{wasysym}
\usepackage{amssymb}
\usepackage{graphicx}
\usepackage{bussproofs}
\usepackage{forest}
\usepackage{caption}
\usepackage[all]{xy}
\usepackage{tikz}
\usetikzlibrary{automata, trees, positioning, arrows}
\usepackage{xparse}
\usepackage{tabularx}
\usepackage{relsize}
\usepackage{adjustbox}
\usepackage{afterpage}
\usepackage{pgfplots}\pgfplotsset{compat=1.9}
\usepackage{subcaption}

\newcommand{\limp}{\rightarrow}

\newcommand{\forces}{\Vdash}

\newcommand{\bigO}{\mathcal{O}}

\newcommand{\cegarboxpp}{\texttt{CEGARBox++}}
\newcommand{\cegartab}{{{\sc CegarTab}}}
\newcommand{\recartab}{{\sc RecarTab}}
\newcommand{\recarbox}{{\tt RECARBox}}
\newcommand{\cegarbox}{\texttt{CEGARBox}}

\newcommand{\minisat}{\texttt{MiniSAT}}

\newcommand{\mosaic}{\texttt{MOSAIC}}

\newcommand{\myk}{K}
\newcommand{\kn}{\ensuremath{\mathrm{K_n}}}

\newcommand{\alc}{ALC}

\usepackage{mathtools}
\usepackage{xspace}
\newcommand{\KSP}{\texttt{K}\kern -.1em \raisebox{-.24em}{\texttt{S}}\kern -.05em \texttt{P}\xspace}
\newcommand{\ksp}{\KSP{}}

\newcommand{\snfmc}{\ensuremath{\textrm{SNF}_{mc}}}
\newcommand{\snfml}{\ensuremath{\textrm{SNF}_{ml}}}

\newcommand{\mcnf}{\snfmc{}}

\newcommand{\wbrex}[1]{#1}
\newcommand{\ccala}[1]{\mcal{C}_{\wbrex{#1}}}

\newcommand{\rosen}{Ros\'en}

\newcommand{\atm}{Atm}
\newcommand{\ag}{Ag}

\newcommand{\bx}[1]{[\, #1 \,]}
\newcommand{\di}[1]{\langle #1 \rangle}

\usepackage{float}

\usepackage{url}
\usepackage{amssymb}
\usepackage{graphicx}
\usepackage{bussproofs}
\usepackage[all]{xy}
\usepackage{tikz}
\usepackage{pgfplots}\pgfplotsset{compat=1.9}
\usepackage{algorithm}
\usepackage{algorithmic}
\newcommand{\wbx}[1]{\bx{#1}}
\newcommand{\wdi}[1]{\di{#1}}

\newcommand{\wdia}{\wdi{\alpha}}
\newcommand{\wbxa}{\wbx{\alpha}}

\newcommand{\wbxb}{\wbx{\beta}}

\newcommand{\vetof}[1]{\mathcal{C}^{cpl}(#1)}

\newcommand{\code}[1]{\texttt{#1}}
\newcommand{\ucore}[1]{UC_{#1}}
\newcommand{\eps}{\epsilon}

\newcommand{\ass}[1]{\mathcal{A}_{#1}}

\newcommand{\cmodel}[1]{\vartheta_{#1}}

\newcommand{\ccpl}[1]{\mathcal{C}^{cpl}_{#1}}
\newcommand{\cdia}[1]{\mathcal{C}^{dia}_{#1}}
\newcommand{\cbox}[1]{\mathcal{C}^{box}_{#1}}
\newcommand{\mcal}[1]{\mathcal{#1}}
\newcommand{\ccal}{\mcal{C}}

\newcommand{\semic}{\ ;\ }

\newcommand{\scsat}{\mbox{\Lightning}SAT}
\newcommand{\scunsat}{\mbox{\Lightning}UNSAT}

\newcommand{\nclap}[1]{\makebox[0pt]{\hss#1\hss}}
\newcommand{\adi}{
  \sbox0{$\lozenge$}
  \usebox0\kern-.5\wd0\nclap{\raisebox{.1ex}{\scalebox{.1}[1]{$-$}}}\kern.5\wd0%
}

\title{Modal CEGAR-tableaux with RECAR and resolution-based SAT-shortcuts}
\author{Rajeev Gor\'e
  \institute{Faculty of Information Technology, Monash University, Australia }
\email{rajeev.gore@monash.edu}
\and
Cormac Kikkert
\institute{Cormac Kikkert Research}
}

\newcommand{\titlerunning}{Modal CEGAR-tableaux with RECAR and resolution-based SAT-shortcuts}
\newcommand{\authorrunning}{Rajeev Gor\'e \& Cormac Kikkert}

\begin{document}



\maketitle
\begin{abstract}
  We investigate two approaches for extending CEGAR-tableaux with SAT-shortcuts using a previously known approach called RECAR
  but also a totally new approach using the modal resolution theorem prover \ksp{} as an oracle.
  Our experiments using our C++ implementation \cegarboxpp{} of CEGAR-tableaux show that:
  (1) \cegarboxpp{} with RECAR SAT-shortcuts is not competitive
  (2) \cegarboxpp{} using \ksp{} to provide SAT-shortcuts is superior to both \cegarboxpp{} and \ksp{},
  particularly on large satisfiable problems.
  As far as we know, this is the first effective integration of SAT, tableaux and resolution methods for modal satisfiability
  which performs better than its parts.



\end{abstract}


\section{Introduction}
\label{sec:intro}

Propositional modal, temporal and description logics are of fundamental importance in logic-based artificial intelligence,
hardware and software verification~\cite{DBLP:conf/icla/Vardi09} and knowledge representation and
reasoning~\cite{baader2008description}.
Thus efficient decision procedures for modal
logic  are an important area of research.

Such research
has not matched the
advances in SAT-solving, until recently, when various authors independently used SAT-solvers
for modal
satisfiability~\cite{pagram-hons,DBLP:conf/gandalf/GeattiGM21,DBLP:conf/hvc/LiZPV15,DBLP:journals/fmsd/LiZPZV19,DBLP:conf/tableaux/GoreK21}
using a technique
originally 
due to
Claessen and \rosen{}~\cite{DBLP:conf/lpar/ClaessenR15} for propositional intuitionistic logic, which 
was, itself, ``inspired'' by the work of Gor\'e et al. on ``modal clause learning'' using binary decision
diagrams~\cite{DBLP:conf/cade/GoreT13}.

We previously~\cite{DBLP:conf/tableaux/GoreK21}
presented a new type of tableaux calculus which uses a SAT-solver and the standard counter-example guided abstraction refinement
(CEGAR) methodology~\cite{DBLP:journals/jacm/ClarkeGJLV03} to search for a rooted Kripke model for a given formula in modal clausal normal form
(mcnf).
We
showed how to modify the initial mcnf to ``compile in'' the axioms T and 4 via a preprocessing stage,
adding a standard loop-check for termination when required. Our
Haskell implementation \cegarbox{} was, overall, the best over the standard benchmarks for the modal logics K, KT
and S4, sometimes by orders of magnitude.
Indeed, the only benchmark where \cegarbox{} did not win were the K-MQBF benchmarks,
where the modal resolution theorem prover \ksp{} solved approximately 50 (out of 1000) more problems in
32 seconds each~\cite{DBLP:conf/tableaux/GoreK21}: see Figure~\ref{fig:mqbf-before}.

When using SAT-solvers as oracles, we usually first create a formula of classical propositional logic (cpl) which is an
approximation of the given (modal) formula $\varphi_0$.  The approximation is an under-approximation $\check{\varphi_0}$ if its
cpl-{UNSAT}isfiability implies the modal unsatisfiability of $\varphi_0$, giving us an \scunsat-shortcut.  The approximation is an
over-approximation $\hat{\varphi_0}$ if its cpl-{SAT}isfiability implies the modal satisfiability of $\varphi_0$, giving us a
\scsat-shortcut.


Here we report on two ways to add \scsat-shortcuts to CEGAR-tableaux: the RECAR approach and 
a totally new approach using the resolution-based solver \ksp{} as an oracle.
Using our new C++ implementation \cegarboxpp{}, our experiments show that:
  (1) \cegarboxpp{} with RECAR \scsat-shortcuts is not competitive; and
  (2) \cegarboxpp{} using \ksp{} to provide \scsat-shortcuts is superior to each, particularly on large satisfiable problems.

\section{Preliminaries: syntax, semantics and modal clausal normal form}

\begin{figure}[t]
\centering
\begin{tabular}[c]{l@{\extracolsep{1em}}cl}
  $M, w \forces \wdia \varphi$ & iff & $\exists v, R_\alpha(w, v)~ \&~ M, v \forces \varphi$
  \\
  $M, w \forces \wbxa\varphi$ & iff & $\forall v, R_\alpha(w, v) \Rightarrow M, v \forces \varphi$
  \\
\end{tabular}
\caption{Semantics of multi-modal logic \kn{} (aka \alc)}
\label{fig:correspondents}
\end{figure}

Let $\ag$ be a non-empty finite set of agent names.
Let $\atm$ be a set of atomic formulae with
$\atm \cap \ag = \emptyset$.
Consider the language
of formulae defined from atoms $p \in \atm$  and
$\alpha \in \ag$ by the BNF grammar:
\begin{center}
  \begin{math}
\varphi \; ::=
\; \bot \; | \; \top \; | \; p \; | \; \neg \varphi
\; | \; \varphi \wedge \varphi \; | \; \varphi \vee \varphi
\; | \; \wbx{\alpha} \varphi \; | \; \wdi{\alpha} \varphi
  \end{math}
\end{center}
Define $(\varphi_1 \rightarrow \varphi_2) := (\neg \varphi_1 \vee \varphi_2)$ and
$(\varphi_1 \leftrightarrow \varphi_2) := ((\varphi_1 \rightarrow \varphi_2) \wedge (\varphi_2 \rightarrow \varphi_1))$.  Let
$\wbxa^0\varphi := \varphi $ and $\wbxa^{n+1}\varphi := \wbxa\wbxa^{n}\varphi$.
The modal-depth of a formula is the maximum number of nested modalities in it.
For a finite set $\Gamma = \{\varphi_1, \cdots , \varphi_k\}$ of formulae, let
$\widehat{\Gamma} = (\varphi_1 \land\cdots\land \varphi_k)$  and let 
$\wbxa\Gamma = \{\wbxa\varphi_1, \cdots , \wbxa\varphi_k\}$.

The Kripke semantics for these multimodal (description) logics use Kripke models consisting of triples
$M := \langle W,$ $\{R_\alpha\}_{\alpha\in\ag},$ $\vartheta\rangle$ over some non-empty set $W$ (of possible worlds), and binary
relations $R_\alpha$ over $W$ for every agent $\alpha\in\ag$, and a valuation $\vartheta(w,p) \subseteq W\times\atm$ telling us
the truth value of each atomic formula $p$ at each world $w \in W$.
Truthhood at world $w$ in model $M$, written
as $M, w \Vdash \varphi$, extends the usual truth-tables for classical propositional logic (cpl): see
Figure~\ref{fig:correspondents}.

A formula $\varphi$ is $\kn$-satisfiable if there is a Kripke model $M$ containing a world
$w$ such that $M, w \forces \varphi$.  A formula $\varphi$ is $\kn$-valid if $\lnot\varphi$ is not
$\kn$-satisfiable.
Formulae $\varphi$ and $\psi$ are logically equivalent if $\varphi \leftrightarrow \psi$ is
\kn-valid, and are
equi-satisfiable if $\varphi$ is
\kn-satisfiable iff $\psi{}$ is \kn-satisfiable.

\subsection{Separated modal clausal normal form for \texorpdfstring{$\kn$}{}}

A \textit{positive literal} is an atomic formula $p$.  A \textit{negative literal} is a negated atomic formula $\lnot p$.  A
\textit{literal} is a positive literal $p$ or else a negative literal $\lnot p$.  We use $a$ to $f$ and $l$ and $r$ for literals.
We use $A$, $B$, $C$ and $D$ for a set of literals.  Let $\bar{l} := \lnot p$ if $l = p$ and $\bar{l} := p$ if $l = \lnot p$ so
that $\bar{\bar{l}} = l$.  A formula is in \textit{negation normal form} (NNF) if it is implication-free and negations appear only in front
of atomic formulae.
An NNF can be created with a linear descent of the formula.

\begin{proposition}
  A formula $\varphi_0$ of modal depth $\kappa$ can be converted into a logically equivalent
  formula $nnf(\varphi_0)$ in NNF which is at most polynomially longer.
\end{proposition}

A modal clause is any formula of one of the following forms~\cite{tseitin1983complexity,mints1988gentzen,DBLP:journals/fuin/GoreN09}:
\begin{description}
\item[\rm cpl-clause:] a formula $(l_1 \land \cdots \land l_i) \limp (r_1 \lor \cdots \lor r_j)$ of literals;
\item[\rm box-clause:] a formula $a \to \wbxa b$ with literals $a$ and $b$ and agent name $\alpha\in\ag$;
\item[\rm dia-clause:] a formula $c \to \wdia d$ with literals $c$ and $d$ and agent name $\alpha\in\ag$.
\end{description}

Using semicolon to indicate a set-union, an arbitrary set $\ccal = (\ccpl{} \semic \cbox{} \semic$ $\cdia{})$, of
modal-clauses can be partitioned into the cpl-clauses $\ccpl{}$ and box-clauses $\cbox{}$ and dia-clauses $\cdia{}$ as defined
above.

Our modal clausal normal form requires notation to succinctly express finite sequences of modalities
such as
$\wbx{\alpha_1}\wbx{\alpha_2}\cdots\wbx{\alpha_k}$ which naturally correspond to all finite paths
$w_1 R_{\alpha_1} w_{2} R_{\alpha_2} \cdots R_{\alpha_k} w_k$ in a Kripke model.
We therefore abuse notation to extend the language of formulae using constructs from regular expressions
even though they are not part of the official syntax.

We use the composition operator ``;'' from regular expressions
with $\eps$ for the empty regular expression obeying $(\eps ; \alpha) = (\alpha ; \eps) = \alpha$.
We let $\alpha^0 = \eps$ and $\alpha^{k>0} = (\alpha ; \alpha^{k-1})$.
We write
$\wbx{\alpha_1 ; \alpha_2 ;$ $\cdots ; \alpha_k}$ instead of
$\wbx{\alpha_1}\wbx{\alpha_2}\cdots\wbx{\alpha_k}$, 
write $\wbx{\alpha^n ; \beta^m}$ for
$\wbxa^n\wbxb^m$ and
write
$w_1 R_{\alpha_1 ; \cdots ; \alpha_k} w_k$ 
for
$w_1 R_{\alpha_1} w_{2} R_{\alpha_2}$ $\cdots R_{\alpha_k} w_k$.
Let 
$\ag^*$
be the set of all finite regular expressions over $\ag$
that use only ``$;$''
 and let $\ag^k$
 be the set of all finite regular expressions of length $k$ over $\ag$
 that use only ``$;$''.

\begin{proposition}
  \label{prop:mcnf}
  A formula $\varphi_0$ of modal depth $\kappa$ can be put into an equi-satisfiable modal clausal normal form
  $r \land \widehat{\mcnf(\varphi_0)}$
  by naming subformulae using new atomic
  propositions, using $r$ as the name for $\varphi_0$,
  and   where each $\ccal_{}$ below is a finite set of modal clauses
  (essentially the $SNF_{ml}$
  of Nalon et al.~\cite{DBLP:conf/ijcai/NalonHD17}
  who prove this proposition  in detail):
  \begin{center}
    \begin{math}
        \begin{array}{lcl}
          \mcnf(\varphi_0) & :=
          & \ccal_{\eps} \semic \bigcup_{\sigma\in\ag^1}^{} \wbx{\sigma}\ccal_{\sigma}
          \semic\cdots\semic
                    \bigcup_{\sigma\in\ag^\kappa}^{} \wbx{\sigma} \ccal_{\sigma}
        \end{array}
      \end{math}
    \end{center}
  \end{proposition}

  \begin{example}
    \label{example-one}
  Consider the negation $\varphi_0 := \lnot(\wbxa(p \limp q) \limp (\wbxa p \limp \wbxa q))$ of the K axiom with modal depth $\kappa =1$.
  Its NNF is
  $\wbxa(\lnot p \lor q) \land \wbxa p \land \wdia\lnot q$.
  Putting $r$ as the name for $\varphi_0$ gives us:
  $\ccala{\eps} := r \limp \wbxa b ~;~ r \limp \wbxa p ~;~ r \limp \wdia\lnot q$ 
  and
  $\ccala{\alpha} := b \limp \lnot p \lor q$ with
  $\mcnf(\varphi_0) = \ccal_\eps{} ~;~ \wbxa \ccal_{\alpha}$.
  \end{example}

Note, the set $\ccal_{\eps}$ excludes $r$ so we must add it explicitly to form 
$(r \semic \mcnf(\varphi_0))$.


\section{SAT-solvers with internal state as CPL-oracles}

CEGAR-tableaux~\cite{DBLP:conf/tableaux/GoreK21}
assume we have access to a SAT-solver 
\code{s}, with internal state,
to which we can add 
a cpl-clause $\varphi$ via
\code{addClause(s, $\varphi$)},
and which accepts a set $\ass{}$ of literals, called unit assumptions,
such as \minisat{}~\cite{een2006minisat}.
Intuitively, we want all literals in
$\ass{}$ to be assigned to true.

When pre-loaded with a
set  $\ccpl{}$ of cpl-clauses 
and 
called with \code{solve(s, $\ass{}$)}, such a SAT-solver returns one of two results:
\begin{description}
\item[\rm $(sat, \cmodel{})$:] if
  $(\ass{}\semic\ccpl{})$ is true under 
  some cpl-valuation 
  $\cmodel{} \supseteq \ass{}$ 
  or-else
\item[\rm $(unsat, \ucore{})$:] if   $\ucore{} \subseteq \ass{}$ is a, not necessarily unique,
  minimal ``unsatisfiable core'' of $\ass{}$
  such that
  $(\ucore{} \semic \ccpl{})$
  is  cpl-unsatisfiable, and hence so is $(\ass{}\semic\ccpl{})$. 
\end{description}
$UC$ itself may be cpl-satisfiable but the term ``unsatisfiable core'' is standard.

Since SAT-solvers handle classical propositional logic, they should return
a cpl-valuation
$\vartheta \subseteq \atm$.
That is, strictly speaking, a cpl-valuation $\vartheta$ is just a set of atomic formulae.
But in concrete applications, it may be more useful for a user to assert
``I know that atomic formula $c$ is false'',
which can be achieved with $\lnot c \in \ass{}$.
Thus, the valuations returned by a SAT-solver are actually a subset of
the set of all literals that appear in 
$(\ass{}\semic\ccpl{})$.

We can then easily extend such a set of literals to a valuation over all atoms
 that appear in an mcnf by putting all missing atoms to false.


\begin{example}
  From our previous example,
  $\varphi_0$ is equi-satisfiable with 
    $(r ~;~ \ccal_\eps{} ~;~ \wbxa \ccal_{\alpha})$.
  Create a SAT-solver
  \code{s0} for modal depth $0$,
  pre-load \code{s0} with $\ccpl{\eps} = \emptyset$, the cpl-part of $\ccala{\eps}$,
  put $\ass{\eps} = \{r\}$ because $r$ must be true at modal depth 0, and
  call
  \code{solve(s0, $\ass{\eps}$)}: it must return
  $(sat, \vartheta_\eps)$ where
  $\vartheta_\eps = \{r\} \supseteq \ass{\eps}$
  since
  (\code{$\ass{\eps}$} ; $\ccpl{\eps}) =  (r \semic \emptyset) = \{r\}$ is true
  under $\vartheta_\eps = \{r \}$.

  Create a SAT-solver
  \code{s1} for modal depth 1,
  pre-load
  \code{s1} with $\ccpl{\alpha} = \{b \limp \lnot p \lor q\}$,
  the cpl-part of $\ccala{\alpha}$,
  put
  $\ass{\alpha} = \{b, p, \lnot q\}$,
  and call
  \code{solve(s1, $\ass{\alpha}$)}: it must return
  $(unsat, \ucore\alpha)$ where
  $\ucore{\alpha} = \{b, p, \lnot q \} \subseteq \ass{\alpha}$
  since 
  (\code{$\ucore\alpha$} ; $\ccpl{\sigma}) =  \{b, p, \lnot q, b \limp \lnot p \lor q\}$
  is cpl-unsatisfiable
  but every proper subset is cpl-satisfiable.

  Finally, 
  add $\lnot r$ to \code{s0} via
  \code{addClause(s0, $\lnot r$)},
  and restart it via
  \code{solve(s0, $\ass{\eps}$)}: it must return 
  $(unsat, \ucore\eps)$ where
  $\ucore\eps = \{r\} \subseteq \ass\eps$
  because 
  $(\ucore{\eps} \semic$ $ (\ccpl{\eps} \semic \lnot r))$ $ = (r \semic \lnot r)$
  is cpl-unsatisfiable but no proper subset is so.
\end{example}

\section{CEGAR-tableaux for multi-modal logic \texorpdfstring{\kn}{} (aka \texorpdfstring{\alc}{})}

\begin{figure}[t]
\begin{algorithm}[H]
    \caption{\cegartab($A$, \code{Trie[$\sigma$]})}
\begin{algorithmic}[1]
  \STATE \COMMENT{Inputs: $A$ is a set of unit assumptions, \\ ~~~~~~\code{Trie[$\sigma$]} is a node in our Trie containing modal clauses and a
    SAT-solver.}
    
	\STATE Let $t_\sigma$  := solve(\code{Trie[$\sigma$].sat}, $A$) \COMMENT{\textit{is $\ccpl{\sigma}$ cpl-satisfiable}}
        \IF{$t_\sigma$ = (unsat, $UC_\sigma$)} 
          \STATE \textbf{return} Unsatisfiable($UC_\sigma$) \COMMENT{\textit{because $\ccala{\sigma}$ is \kn{}-unsatisfiable}}
	\ELSIF{$t_\sigma$ = (sat, $\vartheta_\sigma$)}
        \STATE \COMMENT{\textit{Check box and diamond clauses that fire under classical valuation  $\vartheta_\sigma$}}
	    \FOR {$(c \rightarrow \wdia d) \in $ \code{Trie[$\sigma$].DiaCl} \mbox{ with } $c \in \vartheta_\sigma$ }
	        \STATE Let $B = \{ b \; | \; (a \rightarrow \wbxa b) \in
				\mbox{\code{Trie[$\sigma$].BoxCl}} \mbox{ and } a \in \vartheta_\sigma\}$
                                \STATE \COMMENT{\textit{evaluate the next $\alpha$-successor using the jump rule}}
                    \IF {\cegartab($(d;B)$, \code{Trie[$\sigma$].child($\alpha$)}) = Unsatisfiable($UC_{\sigma;\alpha}$)}
	            \STATE $CS := \{c\} \cup \{ a \; | \; (a \rightarrow \wbxa b) \in \mbox{\code{Trie[$\sigma$].BoxCl}}
                                 \mbox{ and } a \in \vartheta_\sigma \mbox{ and } b \in UC_{\sigma;\alpha}\}$ 
                    \STATE{ Let $\varphi := \bigvee_{l \in CS} \neg l$} 
	            \STATE addClause(\code{Trie[$\sigma$].sat}, $\varphi$)
                    \COMMENT{ \textit{Learn new clause $\varphi := \lnot \bigwedge CS$}}
            \STATE \textbf{return} \cegartab($A$, \code{Trie[$\sigma$]}) \COMMENT{ \textit{apply (restart) }}
	        \ENDIF
        \ENDFOR
            \STATE \textbf{return} Satisfiable \COMMENT{\textit{because every fired diamond is fulfilled}}
	\ENDIF
\end{algorithmic}
\end{algorithm}
\caption{Algorithm of \cegartab{} for multi-modal (description) logic \kn{} (\alc{})}
\label{fig:kn}
\end{figure}

We now describe the CEGAR-tableaux~\cite{DBLP:conf/tableaux/GoreK21} procedure 
extended to multi-modal logic 
\kn{} (aka \alc{})
without local (ABox) and global (TBox) assumptions.

Conceptually, the extension from mono-modal logic K to multi-modal logic \kn{} is easy because there
are no interactions between the different modalities, so we just have $n$ ``copies'' of K, one for
each modality $\wbxa$ for every agent $\alpha \in \ag$.

We give the extension of the main \cegartab{} algorithm from Gor\'e{} and
Kikkert~\cite{DBLP:conf/tableaux/GoreK21} as Algorithm~1 in Figure~\ref{fig:kn}, and make the connection
to the general form of $\mcnf(\varphi_0)$ but with a non-singleton $\ag$:
\begin{center}
  \begin{math}
    \begin{array}{lcl}
          \mcnf(\varphi_0)
          & :=
          & \ccal_{\eps} \semic \bigcup_{\sigma\in\ag^1}^{} \wbx{\sigma}\ccal_{\sigma}
          \semic\cdots\semic
                    \bigcup_{\sigma\in\ag^\kappa}^{} \wbx{\sigma} \ccal_{\sigma}
    \end{array}
  \end{math}
\end{center}

Because \cegartab{} handled only mono-modal logics,
the mcnf of a given formula $\varphi_0$
of modal depth
$\kappa$
was as below for a singleton $\ag = \{\alpha\}$ (say)
:
\[
  \begin{array}{lcl}
\mcnf(\varphi_0)  :=  \ccal_{\eps}
    \semic \wbxa\ccal_{\alpha}
    \semic \wbx{\alpha^2}\ccal_{\alpha^2}
              \semic \cdots 
              \semic \wbx{\alpha^\kappa}\ccal_{\alpha^\kappa}
\end{array}
\]
Conceptually,
these ``modal contexts'' were stored in a linear Trie (a.k.a.\ prefix tree) data structure with
\texttt{Node[0]} initialised to contain $\ccal_{\eps}$
and node
\texttt{Node[i]},
for $i > 0$,
initialised to contain $\ccal_{\alpha^i}$, eliding 
the box-prefix $\wbx{\alpha^i}$ without loss of generality,
but saving quadratic space.

In the multi-modal case, we use a multi-dimensional trie and all trie-edges are labelled by 
an agent name
and the trie-root is labelled with
$\epsilon$.
For a regular expression $\sigma$, we write $\code{Trie}[\sigma]$ for the corresponding \code{TrieNode}.
For example, from the trie-root node $\code{Trie}[\eps] = \ccal_{\epsilon}$
there is a trie-path labelled by an $\alpha$-edge followed by a $\beta$-edge
to the trie-node
$\code{Trie}[\alpha ; \beta] = \ccal_{\alpha ; \beta}$.

In Algorithm~1 in Figure~\ref{fig:kn}, and in the implementation, each
$\ccal_{\sigma}$ is stored at trie-position
$\texttt{Trie}[\sigma]$
in four fields
by partitioning 
it
into its three disjoint components
$\ccal_{\sigma} = (\ccpl{\sigma} \semic \cdia{\sigma} \semic$ $\cbox{\sigma})$
and adding ``next'' pointers
as follows:
\begin{enumerate}
  \item \texttt{Trie$[\sigma]$.sat} is a dedicated SAT-solver for this trie-node pre-loaded with $\ccpl{\sigma}$
  \item \texttt{Trie$[\sigma]$.DiaCl} contains 
    $\cdia{\sigma} :=  \bigcup_{\alpha \in \ag} \{c \limp \wdia d \mid (c \limp \wdia d) \in \ccal_\sigma\}$
  \item \texttt{Trie$[\sigma]$.BoxCl} contains 
     $\cbox{\sigma} := \bigcup_{\alpha \in \ag} \{a \limp \wbxa b \mid (a \limp \wbxa b) \in \ccal_\sigma\}$
  \item \texttt{Trie$[\sigma]$.Child($\alpha$)} is (a pointer to) the node $\mathtt{Trie}[\sigma;\alpha]$, for each $\alpha\in\ag$.
\end{enumerate}

The second and third fields are structured further so that we can select the dia-clauses
and box-clauses for
a particular agent $\alpha \in \ag$
via
\texttt{Trie$[\sigma]$.DiaCl($\alpha$)}
and
\texttt{Trie$[\sigma]$.BoxCl($\alpha$)}, respectively.
We often write ``formula in a \code{TrieNode}'' or even
``$\varphi \in$ \texttt{Trie[$\sigma$]}" to refer to formulae that are stored in one of these fields.

Recall that $r$ names the initial formula but it is not in the root
\code{Trie[$\eps$]}.
Now we simply
start by calling
\cegartab($\{r\}$, \code{Trie[$\eps$]})
as shown in Figure~\ref{fig:kn}.


\begin{example}
Let us continue with our example.
We know that $\varphi_0$ is equi-satisfiable with 
$r ~;~ \ccal_\eps{} ~;~ \wbxa \ccal_{\alpha}$, where
  $\ccpl{\eps} = \emptyset$ and
  $\cdia{\eps} = (r \limp \wdia\lnot q)$ and
  $\cbox{\eps} = (r \limp \wbxa b  ~;~ r \limp \wbxa p)$,
  meaning
  $\texttt{Trie}[\eps].\texttt{sat}$ is  \code{s0}.
  We also have
  $\ccpl{\alpha} = (b \limp \lnot p \lor q)$ and
  $\cdia{\alpha} = \emptyset$ and
  $\cbox{\alpha} = \emptyset$,
  meaning
  \texttt{Trie[$\alpha$].sat} is  \code{s1}.

  \paragraph{Recursion 0.}
  The initial call to
  \cegartab($\{r\}$, \code{Trie[$\eps$]})
  sets $\sigma:= \eps$ and 
  $A := \{r\} = \ass{\eps}$ 
  so line 2 is 
  \code{solve(s0, $\ass{\eps}$)}: it returns
  $(sat, \vartheta_\eps)$ where $\vartheta_\eps = \{r\}$
  so we enter the ``then'' part of line 5.
  At line  7, there is only one ``fired'' dia-clause
  $(r \limp \wdia \lnot q) \in \code{Trie[$\eps$].DiaCl} = \cdia\eps$ with $r \in \vartheta_\eps$,
  so $c :=r$ and $d := (\lnot q)$
  and we know that $\wdia\lnot q$ must be (modally) true at modal depth 0.
  Under
  $\vartheta_\eps = \{r\}$,
  every box-clause in $\cbox{\eps}$ ``fires'',
  thus the set
  $\{\wbxa b, \wbxa p\}$
  is (modally) true at modal depth 0, and hence
  the set
  $B = \{b, p \}$ formed at line 8
  must be classically true at some $R_\alpha{}$-successor for $\lnot q$ at modal depth 1.
  To evaluate this $R_\alpha$-successor, we compute
  $(d;B) = \{b, p, \lnot q\}$ at line 10
  and recurse via
  \cegartab($\{b, p, \lnot q\}$, \code{Trie[$\alpha$].Child($\alpha$)}).

\paragraph{Recursion 1.}
  So 
  $\sigma := \alpha$ and $A := \{b, p, \lnot q\} = \ass\alpha$
  and the relevant SAT-solver is the ``next'' one
  so \code{Trie[$\alpha$].sat} = \code{s1}.
  Line 2 calls
  this SAT-solver, so this is the call
  \code{solve(s1, $\ass{\alpha}$)}: it must return
  $(unsat, \ucore\alpha)$ where
  $\ucore{\alpha} = \{b, p, \lnot q \}$
  so we return Unsatisfiable($\{b, p, \lnot q \}$)
  at line 4 and pop the recursion stack.
  \paragraph{Recursion 0.}
  We return to line 11 
  where $\sigma := \eps$ and so
  $(\sigma ; \alpha) = \alpha$,
  so $UC_{\sigma ; \alpha} = UC_\alpha = \{b, p, \lnot q \}$.
  We know the chosen $\vartheta_\eps$ at modal depth 0 caused a clash
  at modal depth 1 involving  $\{b,p, \lnot q\}$,
  and hence that
  $\{\wbxa b, \wbxa p, \wdia\lnot q\}$ cannot be jointly true at modal depth 0.
  At line 11 we trace the relevant box- and dia-clauses to
  find the conflict set $CS = \{r\}$ at modal depth 0.
  To avoid this ``mistake'', we compute 
  the ``refinement'' $\varphi = \lnot r$ at line 12 of recursion level 0.
  Line 13 is then just \code{addClause(s0, $\lnot r$)}
  so $\ccpl\eps = \{\lnot r\}$.
  The refined $\mcnf(\varphi_0)$
  demands $r$ is false at the root.
  At line 14, we call 
  \cegartab($\{r\}$, \code{Trie[$\eps$]})
  but we stay at the current SAT-solver at modal depth 0
  thereby restarting \code{s0}.
\paragraph{Recursion 1.} 
  Line 2 calls
  \code{solve(s0, $\ass{\eps}$)}: it returns
  $(unsat, \ucore\eps)$ where
  $\ucore\eps = r$
  which returns Unsatisfiable($\{r\}$) 
  to line 14 of Recursion 0:
  our final answer.
\end{example}

\begin{theorem}
  If ~$\mathtt{Trie}$ contains
  $\mcnf(\varphi_0)$ and $r$ names $\varphi_0$ then
  \cegartab($\{r\}$, \code{Trie[$\eps$]})
  terminates, and 
  returns Satisfiable iff 
  $\varphi_0$ is \kn{}-satisfiable.
\end{theorem}
\begin{proof}
  A simple modification of the proofs for mono-modal \myk{} from Gor\'e and Kikkert~\cite{DBLP:conf/tableaux/GoreK21} which proceed
  by induction on the number of restarts.
\end{proof}

\section{CEGAR-tableaux with different (\texorpdfstring{\scsat}{}) shortcuts}\label{chap:satshortcuts}

We previously showed~\cite{DBLP:conf/tableaux/GoreK21} 
that \cegarbox{}
outperformed all solvers we tested on the
extended LWB K-benchmarks, and 3CNF K-benchmarks, mostly by orders of magnitude. However, on the MQBF K-benchmarks, \cegarbox{} loses
to \KSP{} by about 50 problems.
We first determine why and then explore two solutions.

\subsection{Why does \cegarboxpp{} not win on the MQBF benchmarks?}
\begin{figure}[t]
  \centering{}
  \includegraphics[scale=0.5]{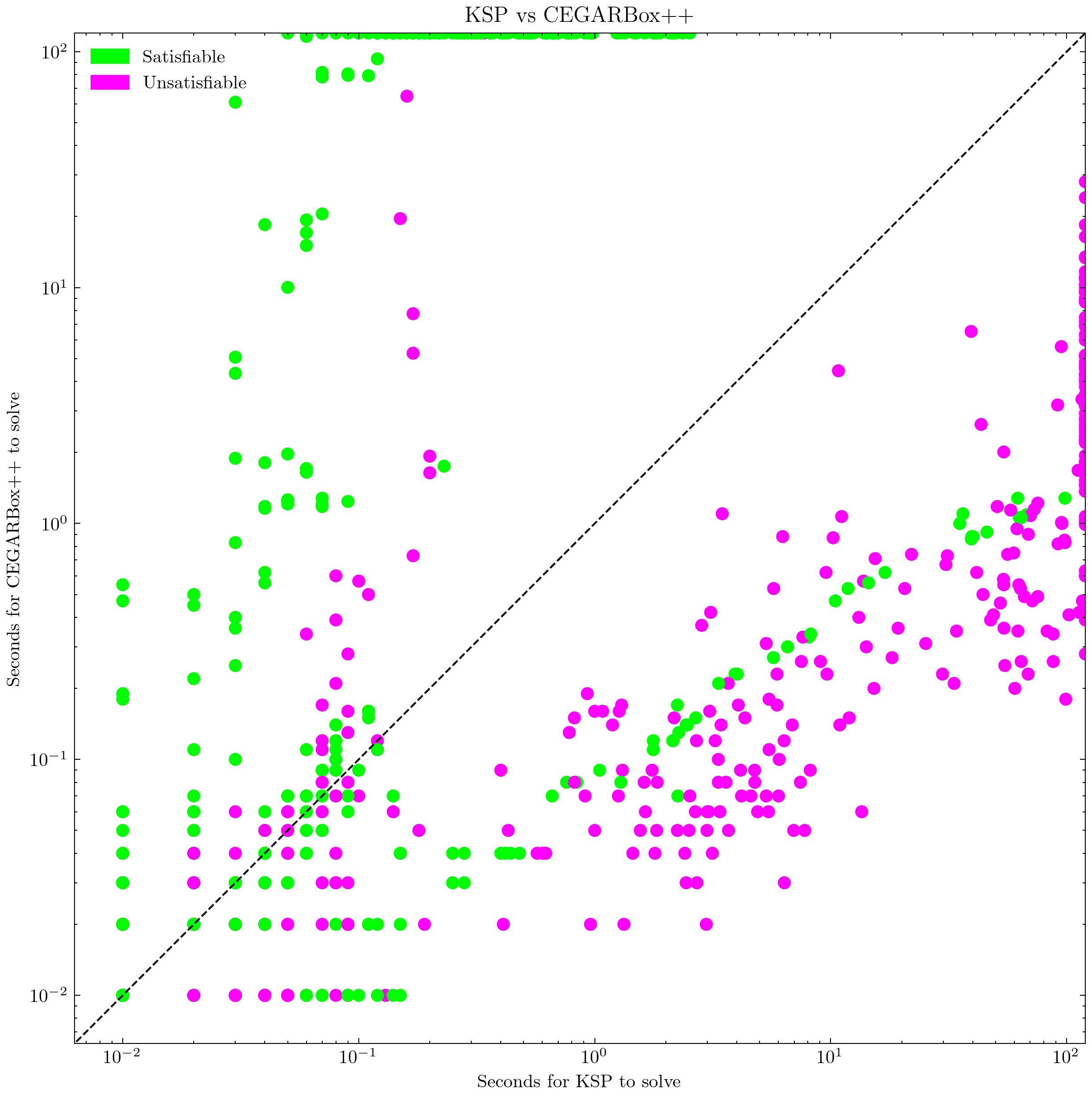}
  \caption{Performance of \cegarboxpp{} and \KSP{} on the MQBF benchmarks. Problems on the top/right edge indicate a timeout
    for \cegarboxpp{}/\KSP{} respectively.}
    \label{fig:mqbf-before}
\end{figure}

Figure~\ref{fig:mqbf-before}
shows
that
\cegarboxpp{} solves all unsatisfiable benchmarks but struggles on the satisfiable ones, while \KSP{} 
solves all satisfiable benchmarks but struggles on the unsatisfiable ones. \KSP{} 
performs better overall as there are
617 satisfiable and 399 unsatisfiable problems in the MQBF benchmarks.

By inspecting the individual (satisfiable) benchmarks that \cegarboxpp{} struggles on, we saw that
\cegarboxpp{} learns few clauses, and instead spends most of its time building a large model.
That is, because \cegarboxpp{} explicitly builds a satisfying model, a huge model can cause it to timeout. This is particularly noticeable in
benchmarks with lots of branching, as
it can result in satisfying models having an exponential number of worlds with respect to
the size of the original benchmark. \KSP{} on the other hand doesn't need to explicitly build a model to deem a formula as
satisfiable, it simply applies rules until \textit{saturation}.

This is a fundamental weakness in not just \cegartab{}, but the CEGAR procedure itself. For a formula $\varphi$, we first
test whether an under-approximation $\check{\varphi}$ consisting of the cpl-clauses of its mcnf is cpl-satisfiable.
If  $\check{\varphi}$ is cpl-unsatisfiable, we can perform an
(\scunsat{})
shortcut by immediately returning UNSAT as
$\varphi$
must also be modally unsatisfiable. However, there is no corresponding
(\scsat)
shortcut.

Our observation is not new.
Brummayer~\cite{DBLP:conf/eurocast/BrummayerB09} devised an under-approximation
for SMT that allows for early termination in the UNSAT case. Wang et al.~\cite{DBLP:conf/fmcad/WangGI07}
used induction to show the existence of large counter-models in model-checking.
These frameworks are domain specific, so it is unlikely their insights would apply to \cegartab{}. However,
for SAT-solving,
Lagniez et al.~\cite{DBLP:conf/ijcai/LagniezBLM17} presented a \textit{general} variant of CEGAR, which allows for
both (\scsat) and (\scunsat) shortcuts.


We first analyse the RECAR approach of
Lagniez et al.~\cite{DBLP:conf/ijcai/LagniezBLM17}. We have previously shown~\cite{DBLP:conf/tableaux/GoreK21}
their implementation was unsound and their benchmarks were
incorrect, but the theory of RECAR is sound
so we incorporate their ideas~\cite{DBLP:conf/ijcai/LagniezBLM17} to create our own variant of RECAR-Tableaux.

We also create our own novel approach
by detecting fixpoints in our under-approximation.
Although it is specific to recursive
CEGAR algorithms,
we show that
it works exceptionally well, allowing
\cegarboxpp{} to solve all problems in the MQBF
benchmark set bar one.
Our approach combines \KSP{} with
\cegarboxpp{},
giving a solver that elegantly
combines SAT, tableaux and resolution methods, which previously were the three competing methods for K-satisfiability.

\subsection{Background on modal resolution and RECAR}
We first give some background
on modal resolution and the RECAR framework.

\subsubsection{Modal clausal resolution (\texorpdfstring{\KSP}{}).}
\label{sec:how-ksp-works}


Recalling the definition of
\snfmc{}, we briefly outline the modal calculus rules of \KSP{}~\cite{DBLP:journals/tocl/NalonDH19} for multi-modal \kn{}, which is a resolution prover that operates on
\snfml{}, a close variant of \snfmc{}.

\begin{proposition}
  \label{prop:snfml}
  A formula $\varphi_0$ of modal depth $\kappa$ can be put into the equi-satisfiable separated normal form with modal layers
  $\snfml(\varphi_0)$ of Nalon et al~\cite{DBLP:conf/ijcai/NalonHD17}, who prove this proposition in detail,
  and where each $\ccal_i$ below is a set of modal clauses:
  \begin{center}
    \begin{math}
        \begin{array}{lcl}
          \snfml(\varphi_0) & := & \ccal_{0} \semic \wbx{}^1 \ccal_{1} \semic \cdots \semic \wbx{}^\kappa \ccal_{\kappa}
        \end{array}
      \end{math}
    \end{center}
  \end{proposition}

\begin{figure}[t]
\begin{align*}
  \text{[LRES]} & \quad \frac{n : D \vee l \quad n : D' \vee \neg l}{n : D \vee D'}
  \hspace*{1cm}
  \text{[MRES]} \quad \frac{n : l_1 \Rightarrow \wbxa \, l \quad n : l_2 \Rightarrow \wdia \neg l}{n : \neg l_1 \vee \neg l_2}
  \\ \\
\text{[GEN1]} & \quad \frac{
\begin{aligned}
n : l'_{1} \Rightarrow \wbxa \, \neg l_1 
, \cdots ,
n : l'_{m} \Rightarrow \wbxa \, \neg l_m ,
n : l' \Rightarrow \wdia \, \neg t ,
n+1 : l_1 \vee \dots \vee l_m \vee t
\end{aligned}
}{
n : \neg l'_{1} \vee \dots \vee \neg l'_{m} \vee \neg l'
                                         }
  \\ \\
  \text{[GEN2]} & \quad \frac{n : l'_{1} \Rightarrow \wbxa \, l_1 \quad n : l'_{2} \Rightarrow \wbxa \, \neg l_1 \quad n : l'_{3} \Rightarrow \wdia \, l_2}{n : \neg l'_{1} \vee \neg l'_{2} \vee \neg l'_{3}}
  \\ \\
\text{[GEN3]} & \quad \frac{
\begin{aligned}
n : l'_{1} \Rightarrow \wbxa \, \neg l_1 
, \cdots ,
n : l'_{m} \Rightarrow \wbxa \, \neg l_m ,
n : l' \Rightarrow \wdia \, l ,
n+1 : l_1 \vee \dots \vee l_m
\end{aligned}
}{
n : \neg l'_{1} \vee \dots \vee \neg l'_{m} \vee \neg l'
}
\end{align*}
\caption{The resolution rules used by \ksp{}~\cite{DBLP:journals/tocl/NalonDH19}.}
\label{fig:ksp-resolution-rules}
\end{figure}

Informally, but not exactly because \ksp{} requires extra clauses that are not required by \cegarboxpp{}, 
  we have
  $\ccal_{l} \approx \bigcup_{\sigma\in\ag^l} \ccal_{\sigma}$, where
  $\ccal_{\sigma}$ is from our $\mcnf(\varphi_0)$.
  That is, each $\ccal_{l}$ is the ``union'' formed by taking a ``horizontal slice'' across layer
  $l = \mid\sigma\mid$ of our ``vertical tree-like'' $\mcnf(\varphi_0)$,
  where $\mid\sigma\mid$ is the length of the Trie-path $\sigma$.
But now, $\snfml(\varphi_0)$ can be implemented as a linear trie, rather than a branching trie as for $\mcnf(\varphi_0)$.
Let us write $n : \psi$ to indicate the clause $\psi$ is stored at the $n$-th node in the (linear) trie,
where the root is element $0$.
Using this notation, the resolution rules used by \KSP{} are presented in Figure~\ref{fig:ksp-resolution-rules}.

They differ from normal resolution as clauses are now labelled by their modal layer.
Applying these rules
will either derive an empty clause, meaning $\varphi_0$ is
\kn{}-unsatisfiable, or terminate after saturation, meaning $\varphi_0$ is
\kn{}-satisfiable.
Observe the following close connection between modal resolution and 
\cegartab{}: all resolvents in \KSP{} are classical clauses as are all learnt clauses in \cegartab{}. We will
use this insight to combine both these 
approaches.

\subsubsection{The RECAR framework (\mosaic{}).}

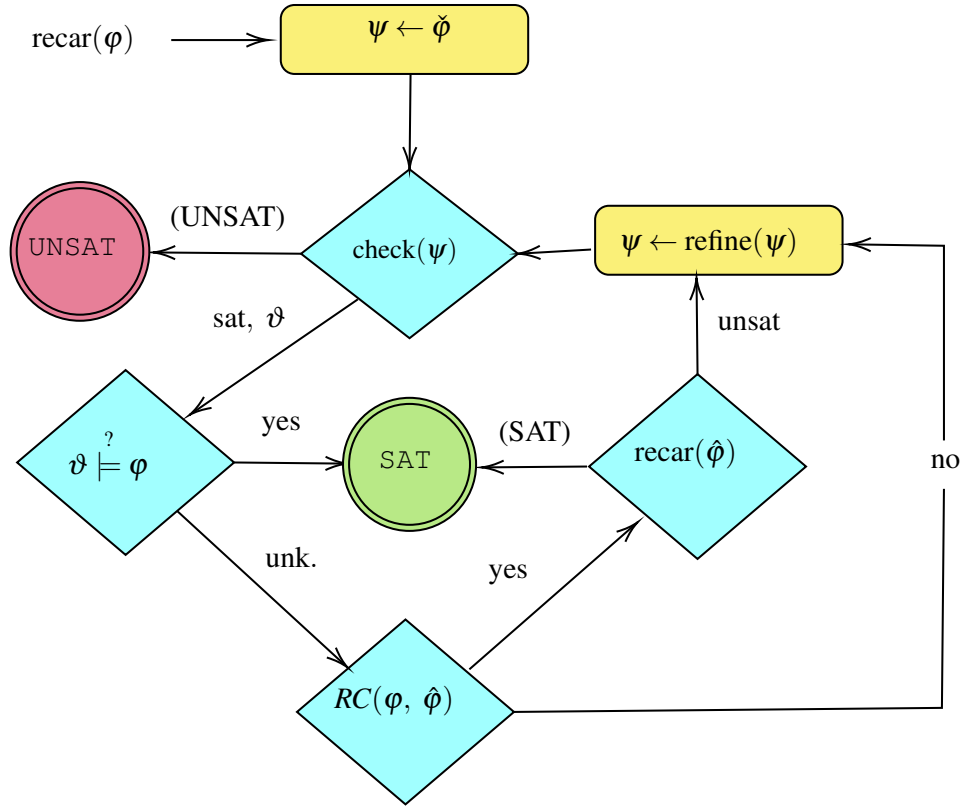
\begin{figure}[t]
  \centering
    \resizebox{0.8\textwidth}{!}{\tikzset{every picture/.style={line width=0.75pt}} 

\begin{tikzpicture}[x=0.75pt,y=0.75pt,yscale=-1,xscale=1]

\draw  [fill={rgb, 255:red, 250; green, 240; blue, 120 }  ,fill opacity=1 ] (227,404.8) .. controls (227,401.04) and (230.04,398) .. (233.8,398) -- (342.2,398) .. controls (345.96,398) and (349,401.04) .. (349,404.8) -- (349,425.2) .. controls (349,428.96) and (345.96,432) .. (342.2,432) -- (233.8,432) .. controls (230.04,432) and (227,428.96) .. (227,425.2) -- cycle ;
\draw  [fill={rgb, 255:red, 160; green, 255; blue, 255 }  ,fill opacity=1 ] (290.5,479) -- (344,520.64) -- (290.5,562.29) -- (237,520.64) -- cycle ;
\draw    (237,520.64) -- (166,520.02) ;
\draw [shift={(164,520)}, rotate = 0.5] [color={rgb, 255:red, 0; green, 0; blue, 0 }  ][line width=0.75]    (10.93,-3.29) .. controls (6.95,-1.4) and (3.31,-0.3) .. (0,0) .. controls (3.31,0.3) and (6.95,1.4) .. (10.93,3.29)   ;
\draw    (291,432) -- (290.52,477) ;
\draw [shift={(290.5,479)}, rotate = 270.61] [color={rgb, 255:red, 0; green, 0; blue, 0 }  ][line width=0.75]    (10.93,-3.29) .. controls (6.95,-1.4) and (3.31,-0.3) .. (0,0) .. controls (3.31,0.3) and (6.95,1.4) .. (10.93,3.29)   ;
\draw  [fill={rgb, 255:red, 161; green, 255; blue, 255 }  ,fill opacity=1 ] (150.5,578) -- (204,623.57) -- (150.5,669.14) -- (97,623.57) -- cycle ;
\draw  [fill={rgb, 255:red, 250; green, 240; blue, 120 }  ,fill opacity=1 ] (382,503.8) .. controls (382,500.04) and (385.04,497) .. (388.8,497) -- (497.2,497) .. controls (500.96,497) and (504,500.04) .. (504,503.8) -- (504,524.2) .. controls (504,527.96) and (500.96,531) .. (497.2,531) -- (388.8,531) .. controls (385.04,531) and (382,527.96) .. (382,524.2) -- cycle ;

\draw    (265,543.14) -- (183.65,599.01) ;
\draw [shift={(182,600.14)}, rotate = 325.52] [color={rgb, 255:red, 0; green, 0; blue, 0 }  ][line width=0.75]    (10.93,-3.29) .. controls (6.95,-1.4) and (3.31,-0.3) .. (0,0) .. controls (3.31,0.3) and (6.95,1.4) .. (10.93,3.29)   ;
\draw    (380,518.57) -- (346,520.53) ;
\draw [shift={(344,520.64)}, rotate = 356.71] [color={rgb, 255:red, 0; green, 0; blue, 0 }  ][line width=0.75]    (10.93,-3.29) .. controls (6.95,-1.4) and (3.31,-0.3) .. (0,0) .. controls (3.31,0.3) and (6.95,1.4) .. (10.93,3.29)   ;
\draw    (204,623.57) -- (258,624.12) ;
\draw [shift={(260,624.14)}, rotate = 180.58] [color={rgb, 255:red, 0; green, 0; blue, 0 }  ][line width=0.75]    (10.93,-3.29) .. controls (6.95,-1.4) and (3.31,-0.3) .. (0,0) .. controls (3.31,0.3) and (6.95,1.4) .. (10.93,3.29)   ;
\draw    (173,415.29) -- (219,415.29) ;
\draw [shift={(221,415.29)}, rotate = 180] [color={rgb, 255:red, 0; green, 0; blue, 0 }  ][line width=0.75]    (10.93,-3.29) .. controls (6.95,-1.4) and (3.31,-0.3) .. (0,0) .. controls (3.31,0.3) and (6.95,1.4) .. (10.93,3.29)   ;
\draw  [fill={rgb, 255:red, 161; green, 255; blue, 255 }  ,fill opacity=1 ] (288.5,701) -- (342,746.57) -- (288.5,792.14) -- (235,746.57) -- cycle ;
\draw    (176,647.57) -- (258.53,723.36) ;
\draw [shift={(260,724.71)}, rotate = 222.56] [color={rgb, 255:red, 0; green, 0; blue, 0 }  ][line width=0.75]    (10.93,-3.29) .. controls (6.95,-1.4) and (3.31,-0.3) .. (0,0) .. controls (3.31,0.3) and (6.95,1.4) .. (10.93,3.29)   ;
\draw  [fill={rgb, 255:red, 161; green, 255; blue, 255 }  ,fill opacity=1 ] (432.5,580) -- (486,625.57) -- (432.5,671.14) -- (379,625.57) -- cycle ;
\draw    (320,725.57) -- (401.49,654.32) ;
\draw [shift={(403,653)}, rotate = 138.84] [color={rgb, 255:red, 0; green, 0; blue, 0 }  ][line width=0.75]    (10.93,-3.29) .. controls (6.95,-1.4) and (3.31,-0.3) .. (0,0) .. controls (3.31,0.3) and (6.95,1.4) .. (10.93,3.29)   ;
\draw    (378,625.57) -- (328,625.98) ;
\draw [shift={(326,626)}, rotate = 359.53] [color={rgb, 255:red, 0; green, 0; blue, 0 }  ][line width=0.75]    (10.93,-3.29) .. controls (6.95,-1.4) and (3.31,-0.3) .. (0,0) .. controls (3.31,0.3) and (6.95,1.4) .. (10.93,3.29)   ;
\draw    (554,612.71) -- (554,516.71) -- (509,516.03) ;
\draw [shift={(507,516)}, rotate = 0.87] [color={rgb, 255:red, 0; green, 0; blue, 0 }  ][line width=0.75]    (10.93,-3.29) .. controls (6.95,-1.4) and (3.31,-0.3) .. (0,0) .. controls (3.31,0.3) and (6.95,1.4) .. (10.93,3.29)   ;
\draw    (342,746.57) -- (553,744.57) -- (554,636.57) ;
\draw    (432.5,580) -- (432.02,534.57) ;
\draw [shift={(432,532.57)}, rotate = 89.4] [color={rgb, 255:red, 0; green, 0; blue, 0 }  ][line width=0.75]    (10.93,-3.29) .. controls (6.95,-1.4) and (3.31,-0.3) .. (0,0) .. controls (3.31,0.3) and (6.95,1.4) .. (10.93,3.29)   ;

\draw (103,406.4) node [anchor=north west][inner sep=0.75pt]    {$\mathrm{recar}( \varphi )$};
\draw  [fill={rgb, 255:red, 231; green, 127; blue, 152 }  ,fill opacity=1 ]  (128, 519.5) circle [x radius= 31.2, y radius= 31.2]   (128, 519.5) circle [x radius= 34.2, y radius= 34.2]  ;
\draw (101,511) node [anchor=north west][inner sep=0.75pt]   [align=left] {{\fontfamily{pcr}\selectfont UNSAT}};
\draw (171,495.4) node [anchor=north west][inner sep=0.75pt]    {($\mathrm{UNSAT}$)};
\draw (267.37,401.4) node [anchor=north west][inner sep=0.75pt]    {$\psi \leftarrow \check{\varphi }$};
\draw (261,511.4) node [anchor=north west][inner sep=0.75pt]  [font=\small]  {$\mathrm{check( \psi )}$};
\draw (120,605.4) node [anchor=north west][inner sep=0.75pt]    {$\vartheta \stackrel{?}{\models } \varphi $};
\draw (192,545.4) node [anchor=north west][inner sep=0.75pt]    {$\mathrm{sat,\ \vartheta \ }$};
\draw  [fill={rgb, 255:red, 184; green, 233; blue, 134 }  ,fill opacity=1 ]  (290.5, 624.5) circle [x radius= 29.81, y radius= 29.81]   (290.5, 624.5) circle [x radius= 32.81, y radius= 32.81]  ;
\draw (265,616) node [anchor=north west][inner sep=0.75pt]   [align=left] {{\fontfamily{pcr}\selectfont  \ SAT \ }};
\draw (216,599.4) node [anchor=north west][inner sep=0.75pt]    {$\mathrm{yes}$};
\draw (546,618.4) node [anchor=north west][inner sep=0.75pt]    {$\mathrm{no}$};
\draw (392.37,506.4) node [anchor=north west][inner sep=0.75pt]    {$\psi \leftarrow \mathrm{refine}( \psi )$};
\draw (252,732.11) node [anchor=north west][inner sep=0.75pt]    {$RC( \varphi ,\ \hat{\varphi }) \ \ $};
\draw (218,663.4) node [anchor=north west][inner sep=0.75pt]    {$\mathrm{unk.}$};
\draw (400,610.4) node [anchor=north west][inner sep=0.75pt]    {$\mathrm{recar}(\hat{\varphi })$};
\draw (333,601.4) node [anchor=north west][inner sep=0.75pt]    {($\mathrm{SAT}$)};
\draw (328,672.4) node [anchor=north west][inner sep=0.75pt]    {$\mathrm{yes}$};
\draw (441,548.4) node [anchor=north west][inner sep=0.75pt]    {$\mathrm{unsat\ }$};

\end{tikzpicture}}
    \caption{The RECAR-under framework~\cite{DBLP:conf/ijcai/LagniezBLM17}.
    }
  \label{fig:pdr-layers}
\end{figure}

RECAR (Recursive Explore and Check Abstraction Refinement) is an extension of the CEGAR approach which introduces a recursive step
to allow for an additional shortcut \cite{DBLP:conf/ijcai/LagniezBLM17}. There are two vaiants, RECAR-under, where the main CEGAR
loop contains an UNSAT shortcut
(\scunsat),
whilst the recursive step allows for a SAT shortcut
(\scsat), and
RECAR-over where the main CEGAR loop contains a
SAT shortcut
(\scsat),
whilst the recursive step allows for an UNSAT
shortcut.

We present just the RECAR-under algorithm here, as the main ``CEGAR'' loop of the RECAR-under framework has
(\scunsat)
shortcuts, like CEGAR-tableaux. This will help motivate our own application of RECAR later but note that the RECAR-over algorithm
can be defined analogously (and is what is used in \mosaic{}~\cite{DBLP:conf/ijcai/LagniezBLM17}). We require the following
assumptions, see Figure~\ref{fig:pdr-layers},
where $RC$ is a Boolean function that determines whether or not a recursive
  call occurs:
\begin{definition}{Assumptions for RECAR-under~\cite{DBLP:conf/ijcai/LagniezBLM17}}
    \begin{enumerate}
        \item Function `check' is a sound, complete and terminating implementation that decides if an input formula is K-unsatisfiable.
            
        \item Function $\mbox{refine}\left(\varphi\right)$ is a function that constrains $\varphi$ with more clauses so that: 
        \begin{enumerate}
            \item If the under-approximation $\check{\varphi}$ is K-unsatisfiable, then so is $\mbox{refine}(\check{\varphi})$.
            \item There exists an $n \in \mathbb{N}$ such that $\mbox{refine}^n(\check{\varphi})$ is K-satisfiable iff $\varphi$ is so.
        \end{enumerate}
    
        \item If the over-approximation $\hat{\varphi}$ is K-satisfiable, then so is $\varphi$.
        \item Let $\mbox{over}(\varphi) = \hat{\varphi}$. There exists $n \in \mathbb{N}$ such that
          $RC(\mbox{over}^n(\varphi), \mbox{over}^{n+1}(\varphi))$ evaluates to false (guaranteeing termination). 
    \end{enumerate}
\end{definition}

\begin{proposition}{}{}
    RECAR-under is sound, complete and terminating~\cite{DBLP:conf/ijcai/LagniezBLM17}.
\end{proposition}

The RECAR framework was implemented in the solver \mosaic{}~\cite{DBLP:conf/ijcai/LagniezBLM17}. Briefly, the
over-approximation involves checking if some formula $\varphi$ is K-satisfiable with $\le n$ worlds via a naive SAT translation, with $\bigO({Atom(\varphi) \times n} + n^2)$ variables, where $Atom(\varphi)$ denotes the number of atoms in $\varphi$. The translation has variables for each world, as well as extra variables indicating which worlds are accessible from each other. Over the course of the algorithm, $n$ will increase until it reaches the bound $n=Atom(\varphi)^{\text{depth}(\varphi)}$ where K-satisfiability can be determined as $\varphi$ is K-satisfiable if and only if it has a model with $n$ worlds~\cite{DBLP:conf/mfcs/Nguyen99}. The under-approximation involves ``removing'' conjuncts in $\varphi$. Thus both the under-approximation and over-approximation involve reasoning about the formula $\varphi$ in its entirety via a SAT-solver. This is exactly what CEGAR-tableaux aim to avoid by using a SAT-solver to only determine worlds (and not models), thus requiring only $\bigO(n)$ variables in a SAT-solver. 

\subsection{\recartab{}: Extending CEGAR-tableaux with RECAR}

We now extend CEGAR-tableaux with RECAR
(\scsat)-shortcuts.

First, we must take some liberties when applying RECAR to CEGAR-tableaux. Note our under-approximation
$\check{\varphi} := \vetof{\varphi}$ is a classical logic formula, whilst $\varphi$ is a modal logic formula. As the algorithm
progresses, $\check{\varphi}$ will be refined with classical clauses only, and so whilst there exists an $n$ where
$\mbox{refine}^n(\check{\varphi})$ is K-satisfiable iff $\varphi$ is so,
we are unable to detect when we have reached this $n$. This
differs from \mosaic{}, where the upper bound $n = UB(\varphi)$
can be pre-calculated by recursively counting the number of diamond sub-formulae in $\varphi$.
Second, CEGAR-tableaux already have a recursive
step because they are based on recursive CEGAR loops, unlike RECAR which is based on just one CEGAR loop.

These differences arise because \cegartab{} considers only one world of the Kripke model at a time, as opposed to naive
SAT-translation approaches that translate the whole input modal formula. We believe that keeping worlds separate is
one of the main factors of \cegartab{}'s success, and so we keep the under-approximation as is. Thus, instead of implementing a
faithful RECAR algorithm, we shall instead perform RECAR on each world. Our under-approximation is unchanged, but our
over-approximation is best motivated by an example.

\begin{figure}[H]
\begin{algorithm}[H]
    \caption{\recartab($A$, \code{Trie[$\sigma$]})}
\begin{algorithmic}[1]
  \STATE \COMMENT{Inputs: $A$ is a set of unit assumptions, \\ ~~~~~~\code{Trie[$\sigma$]} is a node in our Trie containing modal clauses and a
    SAT-solver.}
    
\STATE Let $t_\sigma$  := solve(\code{Trie[$\sigma$].sat}, $A$) \COMMENT{\textit{is $\ccpl{\sigma}$ cpl-satisfiable}}
        \IF{$t_\sigma$ = (unsat, $UC_\sigma$)} 
          \STATE \textbf{return} Unsatisfiable($UC_\sigma$) \COMMENT{\textit{because $\ccala{\sigma}$ is \kn{}-unsatisfiable}}
\ELSIF{$t_\sigma$ = (sat, $\vartheta_\sigma$)}
    \STATE Let $L = \{ l \; | \; (c \rightarrow_l \wdia d) \in \code{Trie[$\sigma$].DiaCl} \mbox{ and } c \in \vartheta_\sigma\}$
    \STATE \COMMENT{\textit{Check box and diamond clauses that fire under classical valuation  $\vartheta_\sigma$}}
   \FOR {$i \in L$ }
       \STATE Let $D = \{ d \; | \; (c \rightarrow_l \wdia d) \in
            \mbox{\code{Trie[$\sigma$].DiaCl}} \mbox{ and } c \in \vartheta_\sigma \mbox{ and } l = i\}$
       \STATE Let $B = \{ b \; | \; (a \rightarrow \wbxa b) \in
\mbox{\code{Trie[$\sigma$].BoxCl}} \mbox{ and } a \in \vartheta_\sigma\}$
                                \STATE \COMMENT{\textit{evaluate the next $\alpha$-successor using the jump rule on all fired diamonds}}
                        \IF {\recartab($D \cup B$, \code{Trie[$\sigma$].child($\alpha$)}) = Unsatisfiable($UC_{\sigma;\alpha}$)}
                        \STATE Let $D_{UC} := UC_{\sigma;\alpha} \cap D$
                        \COMMENT{ \textit{How many diamonds in conflict set?}}
                        \IF{$|D_{UC}| = 1$}
                        \STATE \COMMENT{\textit{Normal clause learning}}
                        \STATE Find $c \rightarrow_l \wdia d$ where $d \in D_{UC}$

                        \STATE $CS := \{c\} \cup \{ a \;|\; (a \rightarrow \wbxa b) \in \mbox{\code{Trie[$\sigma$].BoxCl}},\; a \in \vartheta_\sigma,\; b \in UC_{\sigma;\alpha} \}$ 
                        \STATE Let $\varphi := \bigvee_{l \in CS} \neg l$ 
                        \STATE addClause(\code{Trie[$\sigma$].sat}, $\varphi$)
                        \COMMENT{ \textit{Learn new clause $\varphi := \lnot \bigwedge CS$}}

                    \ELSIF{$|D_{UC}| < 1$}
                        \STATE \COMMENT{\textit{No diamond in conflict — pick random $d \in D$ and simulate normal clause learning}}
                        \STATE Pick arbitrary $d \in D$
                        \STATE Find $c \rightarrow_l \wdia d$

                        \STATE $CS := \{c\} \cup \{ a \;|\; (a \rightarrow \wbxa b) \in \mbox{\code{Trie[$\sigma$].BoxCl}},\; a \in \vartheta_\sigma,\; b \in UC_{\sigma;\alpha} \}$ 
                        \STATE Let $\varphi := \bigvee_{l \in CS} \neg l$ 
                        \STATE addClause(\code{Trie[$\sigma$].sat}, $\varphi$)
                    \ELSE
                        \STATE \COMMENT{ \textit{(\scsat) shortcut failed: multiple conflicting diamonds}}
                        \FOR{$(c \rightarrow_l \wdia d) \in \mbox{\code{Trie[$\sigma$].DiaCl}}$ \mbox{ with } $c \in \vartheta_\sigma,\; l = i,\; d \in UC_{\sigma;\alpha}$}
                            \STATE Replace $l$ with a fresh label $l'$ in $c \rightarrow_l \wdia d$
                        \ENDFOR
                    \ENDIF
                \STATE \textbf{return} \recartab($A$, \code{Trie[$\sigma$]}) \COMMENT{ \textit{apply (restart) }}
       \ENDIF
        \ENDFOR
            \STATE \textbf{return} Satisfiable \COMMENT{\textit{because every fired diamond is fulfilled}}
\ENDIF
\end{algorithmic}
\caption{Algorithm of \recartab{}
  for incorporating (\scsat) shortcuts}
\label{recar-tab}
\end{algorithm}
\end{figure}

Suppose we are creating a model with depth $n$, and each world is `firing' two diamond clauses. This results in creating
$\bigO(2^n)$ worlds, and as \cegartab{} must iterate over each of these worlds, this will clearly timeout for large $n$. Instead,
let us create an over-approximation that will aim to reduce the number of successors for each node. Instead of creating two worlds
for $\wdia p \land \wdia q$, one for $p$ and one for $q$, we will attempt to create one $\alpha$-successor with ($p \semic q$). If
this fails we will get a conflict $A'$, and if $\{p, q\} \subseteq A'$, we know we must separate these diamond clauses. Otherwise,
we have just managed to create a model for $w_0$ that has one successor, not two. If every world succeeds in grouping diamond
clauses we can create a model with $\bigO(n)$ worlds instead of $\bigO(2^n)$ worlds, as formalised next.

\subsection{\recartab{}: the RECAR-tableaux over-approximation algorithm}

We begin by annotating each dia-clause with a `label', initialised to 1: i.e.\ we replace $a \limp \wdia b$ with
$a \limp_1 \wdia b$, for all dia-clauses. Now, we create a sole successor for each fired \textit{label}, instead of each fired
\textit{dia-clause}. The number of labels will increase as the algorithm progresses so let the set of labels be $L$.


\begin{theorem}
  Algorithm~\ref{recar-tab} for \recartab{} is a sound and complete decision procedure for \kn{}.
\end{theorem}
\begin{proof}~
\begin{description}
\item[\rmfamily Termination:] Separating diamond clauses can happen only finitely many times. Thus for RECAR-tableaux
  to not terminate, it must either jump or restart infinitely, but both are impossible by the
  proofs for CEGAR-tableaux~\cite{DBLP:conf/tableaux/GoreK21}.
\item[\rmfamily Soundness:] Any clause learned by \recartab{} is of the form $\lnot c \lor \lnot a_1 \lor \cdots \lor \lnot a_m$
  as in line 17 and 24 and is a logical consequence of the discovery that for some label $l$, the set
  $\{ c, a_1, \cdots , a_m, c \limp_l \wdia d , a_1 \limp_l \wbxa b_1, \cdots , a_m \limp_l \wbxa b_m\}$
  is K-unsatisfiable just as they are in \cegartab{},
  so it is sound to learn such a clause. But when the conflict set contains multiple diamonds, as in line 28, we separate these
  conflicting diamonds using different labels, and restart, which will reduce the cardinality of $L$ at line 6.
  This corresponds to finding that 
  $(\wdia D; \wbxa B)$
  is K-unsatisfiable
  and then attempting to K-satisfy
  $(\wdia D'; \wbxa B)$
  for some $D' \subset D$
  instead.
  The move is sound because we maintain the requirement for (\scsat)-shortcuts that
  if $(\wdia D'; \wbxa B)$ is K-satisfiable then so is
  each individual set
  $(\wdia d' ; \wbxa B)$ for all $d' \in D'$. 
  In the worst case, this
  will lead to every dia-clause having its own label, meaning that \recartab{} will just simulate \cegartab{}, which we know to
  be sound.
  
\item[\rmfamily Completeness] Suppose \recartab{} returns Satisfiable and proceed by induction on the number of jumps on the
  maximal sequence of jumps.  If there are no jumps the Kripke model is just a dead-end $w_\eps$ with
  $\vartheta_\eps$. Otherwise, consider the first jump from this root world on this sequence for any $i \in L$.
  The fired diamonds consist of some
  set $\wdia D = \{\wdia d_1 , \cdots , \wdia d_{n > 0} \}$ while the fired boxes are
  $\wbxa B = \{\wbxa b_1, \cdots , \wbxa b_{m \geq 0}\}$. We jump
  to a child containing the assumptions
  $(D ; B) = \{d_1 , \cdots , d_n, b_1, \cdots , b_m\}$ and the recursive call at line~12 must have
  returned Satisfiable so this child $w_\alpha^{i}$ must be
  K-satisfiable by the induction hypothesis.  Thus
  $(\wdia D ; \wbxa B)$ is K-satisfiable at $w_\eps$
  in the K-model obtained by putting
  $w_\eps R_\alpha w_\alpha^i$. But then so must each set
  $(\wdia d ; \wbxa B)$ for every $d \in D$ since
  $\wdia( p \land q) \limp \wdia p \land \wdia q$ is K-valid.
  Thus each $i \in L$ generates such an $R_\alpha$-child for its corresponding
  diamond jump 
  and the for-loop over these $i \in L$ at line 8 succeeds
  only when all such children return Satisfiable, meaning K-satisfiable.
  Thus every $(\wdia d ; \wbxa B)$ fired by $\vartheta_\eps$ has a K-satisfiable witness, giving a K-model for
  $\mcnf(\varphi_0)$.
\end{description}
\end{proof}


We report on experiments on RECAR-Tableaux at the end of this section, but let us first propose our own approach to achieving
(\scsat) shortcuts.

\subsection{(\texorpdfstring{\scsat}{}) shortcuts via fixpoint detection using \texorpdfstring{\KSP}{}}
One of the main issues with \recartab{} is that it explicitly generates the model, meaning it suffers from the same
shortcomings as regular CEGAR-tableaux. Sure, it uses heuristics to find smaller models, but there is no guarantee these smaller models
exist,
and this ``optimism'' might lead to wasted work.

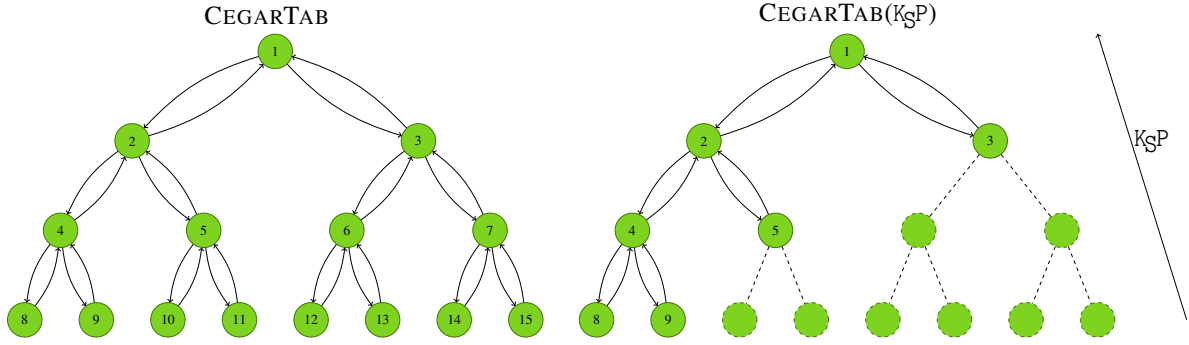
\begin{figure}[t]
  \centering
  \resizebox{\textwidth}{!}{\begin{tikzpicture}[
  every node/.style = {
    circle,
    draw = {rgb, 255:red, 65; green, 117; blue, 5 },
    fill = {rgb, 255:red, 126; green, 211; blue, 33 },
    minimum size = 2.5em
  },
  skipped/.style={
    dashed,
    draw,
    fill=none
  },
  textnode/.style={
    text width=2cm,
    draw=none,
    fill=none
  },
  mycolor/.style={draw = {rgb, 255:red, 0; green, 0; blue, 0 }},
  level distance = 2.5cm,
  sibling distance = 1cm,
  level 1/.style = {sibling distance=8cm},
  level 2/.style = {sibling distance=4cm},
  level 3/.style = {sibling distance=2cm},
  level 4/.style = {sibling distance=1cm},
]

\node (1) {1}
  child {node (2) {2} edge from parent[draw=none]
    child {node (4) {4} edge from parent[draw=none]
      child {node (8) {8} edge from parent[draw=none]}
      child {node (9) {9} edge from parent[draw=none]}
    }
    child {node (5) {5} edge from parent[draw=none]
      child {node (10) {10} edge from parent[draw=none]}
      child {node (11) {11} edge from parent[draw=none]}
    }
  }
  child {node (3) {3} edge from parent[draw=none]
    child {node (6) {6} edge from parent[draw=none]
      child {node (12) {12} edge from parent[draw=none]}
      child {node (13) {13} edge from parent[draw=none]}
    }
    child {node (7) {7} edge from parent[draw=none]
      child {node (14) {14} edge from parent[draw=none]}
      child {node (15) {15} edge from parent[draw=none]}
    }
  };
\node[right=15cm of 1] (1b) {1}
  child {node (2b) {2}  edge from parent[draw=none]
    child {node (4b) {4} edge from parent[draw=none]
      child {node (8b) {8} edge from parent[draw=none]}
      child {node (9b) {9} edge from parent[draw=none]}
    }
    child {node (5b) {5} edge from parent[draw=none]
      child [skipped] {node (10b) {} }
      child [skipped] {node (11b) {} }
    }
  }
  child {node (3b) {3} edge from parent[draw=none]
    child [skipped] {node (6b) {}
        child [skipped] {node (12b) {}}
        child [skipped] {node (13b) {}}
    }
    child [skipped] {node (7b) {}
        child [skipped] {node (14b) {}}
        child [skipped] {node (15b) {}}
    }
  };
\foreach \from/\to in {1/2,2/4,4/8,8/4,4/9,9/4,4/2,2/5,5/10,10/5,5/11,11/5,5/2,2/1,1/3,3/6,6/12,12/6,6/13,13/6,6/3,3/7,7/14,14/7,7/15,15/7,7/3,3/1} {
    \draw[->, mycolor, thick, bend right=15] (\from) to (\to); 
}

\foreach \from/\to in {1b/2b,2b/4b,4b/8b,8b/4b,4b/9b,9b/4b,4b/2b,2b/5b,5b/2b,2b/1b,1b/3b,3b/1b} {
    \draw[->, mycolor, thick, bend right=15] (\from) to (\to); 
}

\coordinate[right=2cm of 15b] (BottomPoint);
\coordinate[right=2cm of 15b, above=7.5cm of 15b] (TopPoint);
\draw[->, thick, mycolor] (BottomPoint) -- (TopPoint);
\node[textnode] at (25, -2.5) {\huge \ksp{}};
\node[textnode] at (-1, 1) {\huge \cegartab{}};
\node[textnode] at (14.5, 1) {\huge \cegartab{(\ksp{})}};
\end{tikzpicture}}
  \caption{A comparison of the search process of \cegartab{} and \cegartab{(\ksp{})}. \cegartab{} evaluates worlds with a
      pre-order traversal, and in the meantime, \ksp{} saturates layers bottom-up.
      So \cegartab{(\ksp)} generates no successors for ``fixpoint'' worlds $5$ and $3$ giving
      scenarios where \cegartab{(\ksp{})} evaluates a linear number of worlds,
      whilst \cegartab{} evaluates an exponential number.}
  \label{fig:planning-to-sat}
\end{figure}

We propose a different method of finding (\scsat) shortcuts. Instead of explicitly creating a model, we simply check if
one exists by searching for fixpoints in the \cegartab{} under-approximation where, intuitively, a fixpoint is a set
of clauses which
will never be refined with any new clauses by \cegartab{}.

\begin{lemma}[Fixpoints are Satisfiable]\label{lem:fixpoint-correctness}
  If $\ass{\sigma}$ is a set of unit assumptions, and 
  \cegartab{} will not refine 
  $\ccpl{\sigma}$
  with any clauses in the future, then $\ass{\sigma} \semic \ccpl{\sigma}$ is classically satisfiable iff
  $\ass{\sigma} \semic \ccpl{\sigma} \semic \cbox{\sigma} \semic \cdia{\sigma}$ is \kn-satisfiable.
\end{lemma}
\begin{proof}
  If $\ass{\sigma} \semic \ccpl{\sigma} \semic \cbox{\sigma} \semic \cdia{\sigma}$ is \kn-satisfiable then its subset
  $\ass{\sigma} \semic \ccpl{\sigma}$ is clearly classically satisfiable.  If
  $\ass{\sigma} \semic \ccpl{\sigma}$ is classically satisfiable then \cegartab{$(\mathcal{A}, Trie[\sigma])$} will find
  a classical valuation, then iterate over diamonds. For
  a contradiction, suppose a diamond jump returns Unsatisfiable. Then we learn a clause to add to $\ccpl{\sigma}$, contradicting that
  $\ccpl{\sigma}$ is a fixpoint. Thus all diamonds must be fulfilled, and so \cegartab{$(\mathcal{A}, Trie[\sigma])$} returns
  Satisfiable. By the correctness of \cegartab{}, the set 
  $\ass{\sigma} \semic \ccpl{\sigma} \semic \cbox{\sigma} \semic \cdia{\sigma}$ is \kn-satisfiable.
\end{proof}

Lemma~\ref{lem:fixpoint-correctness} allows
\cegartab{} to skip the potentially expensive process of generating a model if we know there
are no clauses left to learn. This helps when a satisfying model is large, but ``easy'' to produce.
But how to determine
whether $\ccpl{\sigma}$ is a fixpoint? To this end we use \KSP{}, as explained next.

\subsection{CEGAR-tableaux with a resolution-based oracle}

Recall that \ksp{} uses \snfml{} rather than \mcnf{} and that \snfml{} leads to a linear Trie data structure, as explained previously.
Nevertheless, Lemma~\ref{lem:fixpoint-correctness} still applies.

\begin{theorem}[\ksp{} produces \cegartab{} fixpoints]
  Suppose that a
  \textbf{saturated} derivation from $\snfml(\varphi_0)$
  by
  \KSP{} results in a new set $\ccal_{l}'$ of modal clauses (at layer $l$).
  Then $\ccal_{l}'$ is a fixpoint of \cegartab{}.
\end{theorem}
\begin{proof}
  If $\ccal_{l}'$ is not a \cegartab{} fixpoint, the restart rule is applicable. Then the learnt clause
  corresponds exactly to one of the [GEN] resolution rule applications from Figure~\ref{fig:ksp-resolution-rules} by \KSP{},
  so the set of clauses is
  not saturated: contradiction.
\end{proof}

  \begin{figure}[t]
\begin{algorithm}[H]
    \caption{\cegartab(\KSP)($A$, \code{Trie[$l$]})}
\begin{algorithmic}[1]
  \STATE \COMMENT{Inputs: $A$ is a set of unit assumptions,
    \\ ~~~~~~\code{Trie[$l$]} contains the set $\ccal_{l}$ of modal clauses  from $\snfml(\varphi_0)$  and a
    SAT-solver.}
    
    \STATE \COMMENT{\textit{Merge Step: Check if external KSP process has saturated layer $l$ wrt layer $l+1$ (bottom-up)}}
    \IF{exists KSP-output-file for layer $l$}
        \STATE Add new KSP clauses to \code{Trie[$l$].sat}
        \STATE $Fixpoint := \textbf{true}$
    \ENDIF
    
    \STATE Let $t_l$  := solve(\code{Trie[$l$].sat}, $A$) \COMMENT{\textit{is $\ccpl{l}$ cpl-satisfiable}}
        \IF{$t_l$ = (unsat, $UC_l$)} 
          \STATE \textbf{return} Unsatisfiable($UC_l$) \COMMENT{\textit{because $\ccala{l}$ is \kn{}-unsatisfiable}}
    \ELSIF{$t_l$ = (sat, $\vartheta_l$)}
        \IF{$Fixpoint$ = \textbf{true}}
            \STATE \textbf{return} Satisfiable \COMMENT{\textit{Fixpoint reached by KSP: skip child generation}}
        \ENDIF

        \STATE \COMMENT{\textit{Check box and diamond clauses that fire under classical valuation  $\vartheta_l$}}
        \FOR {every $(c \rightarrow \wdia d) \in $ \code{Trie[$l$].DiaCl} \mbox{ with } $c \in \vartheta_l$ }
            \STATE Let $B = \{ b \; | \; (a \rightarrow \wbxa b) \in
                \mbox{\code{Trie[$l$].BoxCl}} \mbox{ and } a \in \vartheta_l\}$
                                \STATE \COMMENT{\textit{evaluate the next layer using the jump rule on $\alpha$}}
                    \IF {\cegartab($(d;B)$, \code{Trie[$l+1$])}) = Unsatisfiable($UC_{l+1}$)}
                \STATE $CS := \{c\} \cup \{ a \; | \; (a \rightarrow \wbxa b) \in \mbox{\code{Trie[$l$].BoxCl}}
                                 \mbox{ and } a \in \vartheta_l \mbox{ and } b \in UC_{l+1}\}$ 
                    \STATE{ Let $\varphi := \bigvee_{l \in CS} \neg l$} 
                \STATE addClause(\code{Trie[$l$].sat}, $\varphi$)
                    \COMMENT{ \textit{Learn new clause $\varphi := \lnot \bigwedge CS$}}
            \STATE \textbf{return} \cegartab($A$, \code{Trie[$l$]}) \COMMENT{ \textit{apply (restart) }}
            \ENDIF
        \ENDFOR
            \STATE \textbf{return} Satisfiable \COMMENT{\textit{because every fired diamond is fulfilled}}
    \ENDIF
\end{algorithmic}
\end{algorithm}
\caption{Algorithm of \cegartab(\KSP) on $\snfml$ showing the integration of fixpoint detection.}
\label{fig:kn-ksp}
\end{figure}

Our algorithm for ``fixpoint detecting'' CEGAR-Tableaux is 
\cegartab{(\ksp{})} but our implementation is \cegarboxpp{(\ksp{})}. For simplicity, \cegarboxpp{(\ksp{})} is
\textbf{multithreaded}, running \cegarboxpp{} and \ksp{} independently, with naive communication via file
writing/reading. We thank Cl{\'{a}}udia Nalon for adjusting \ksp{} to make this possible. The approach
is as follows where both \cegarboxpp{} and \ksp{} use the linear Trie for \snfml{}, not the tree-like Trie for \mcnf{},
see Figure~\ref{fig:planning-to-sat} and Figure~\ref{fig:kn-ksp}:
\begin{enumerate}
    \item Run \ksp{} on an input formula $\varphi_0$, and have it print out $\snfml{(\varphi_0)}$;
    \item \ksp{} then performs (layered) resolution, from deepest ($\kappa$) to shallowest layers ($0$), and after fully
      saturating a layer, prints out all classical clauses it has learned;
    \item \cegartab{} runs independently on $\snfml{(\varphi_0)}$, from shallowest ($0)$ to deepest ($\kappa$) layers, but when it
      receives clauses from \ksp{} at layer $l$, it adds them to $\ccpl{l}$ (at modal layer $l$), marks the layer as a fixpoint
      and ignores the jump rule at that layer (sound by Lemma~\ref{lem:fixpoint-correctness}).
\end{enumerate}

There are some subtle points in this integration which deserve noting. At any particular time, \ksp{} is
  saturating some layer $l$ with respect to layer $l+1$ using the resolution rules and adding cpl-clauses to layer $l$ as
  required to avoid the jumps that would cause a clash at layer $l+1$.
  If \cegarboxpp{} arrives at layer $l$ from layer $i-1$ \textit{before}
  before \ksp{} has saturated layer $l$ wrt layer $l+1$, then 
  the flag  \textit{Fixpoint} will be false and line 12 will be skipped,
  so \cegarboxpp{(\ksp)} mimics \cegarboxpp{} at layer $l$. Else, if \ksp{} has already saturated layer
  $l$ wrt layer $l+1$, and printed out new cpl-clauses, then 
  \textit{Fixpoint} is put to true at line 5.
  \cegarboxpp{} therefore
  imports the new cpl-clauses at line~4 and calls the SAT-solver at layer $l$ at line 7 with the assumptions $A$.
  That is, \cegarboxpp{}
  checks whether the assumptions $A$ demanded by the jump from
  layer $l-1$ to layer $l$ are jointly cpl-satisfiable wrt the old cpl-clauses of layer $l$ (inside the SAT-solver at layer $l$)
  plus these new cpl-clauses,
  But surely, \ksp{} has already told us that layer $l$ is satisfiable, so why do we need this check?
  Not quite:
  it is perfectly possible that the new cpl-clauses make
  layer $l$ cpl-unsatisfiable under assumptions $A$ as demanded by layer $l-1$,
  since \ksp{} has not checked this possibility (yet).
  Thus the call to the SAT-solver at line 7 is essential in both cases.

At worst, this algorithm \textit{should} be as strong as \cegartab{}.
As we shall see, the combination works surprisingly well, despite the rather naive integration.

\subsection{Experimental evaluation of the two new approaches}
\begin{figure}[t]
  \centering{}
  \includegraphics[width=0.8\textwidth]{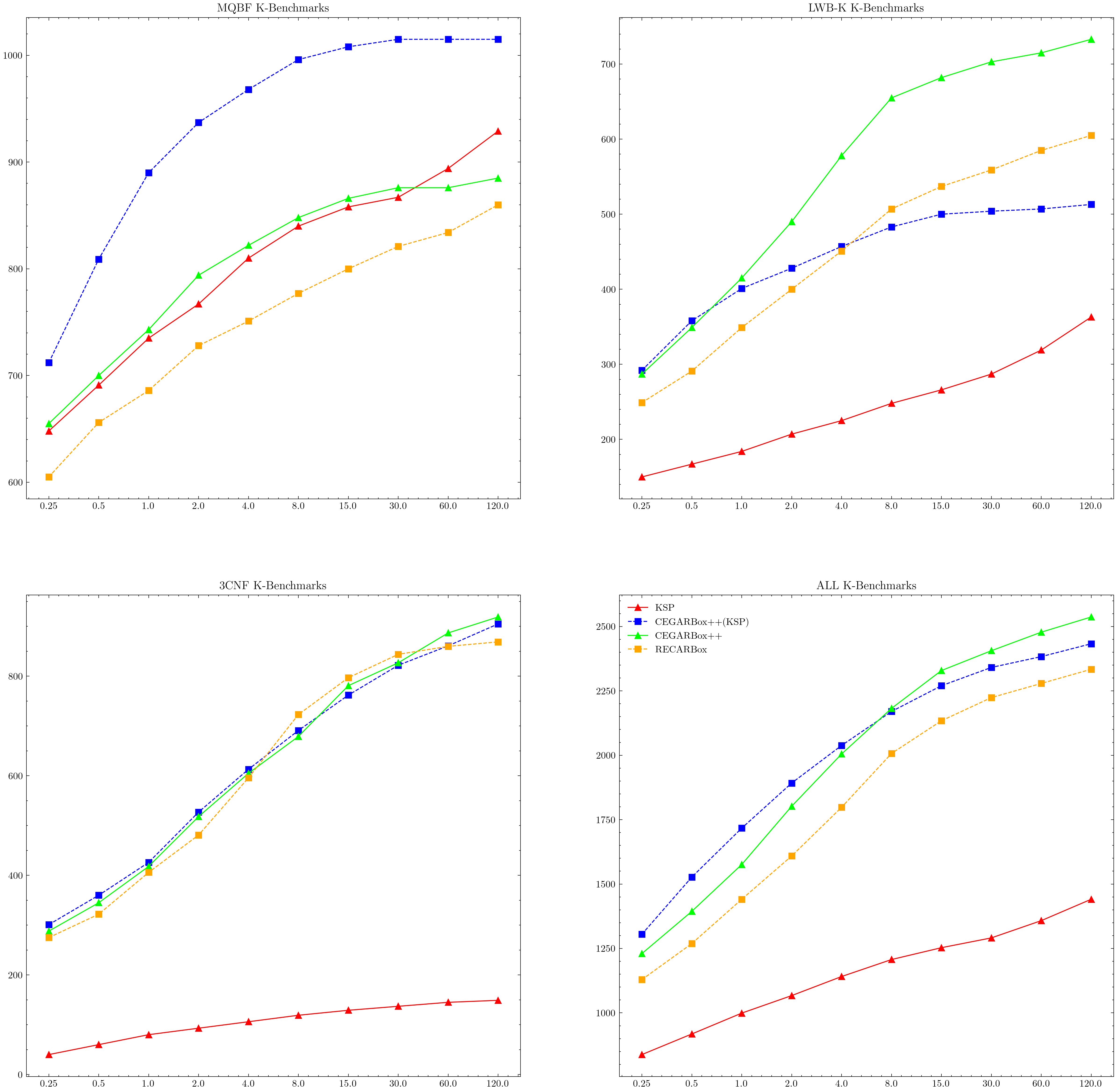}
  \caption{Performance of solvers on the standard K benchmarks}
    \label{fig:all-after}
\end{figure}

\begin{figure}[t]
  \centering{}
  \includegraphics[scale=0.5]{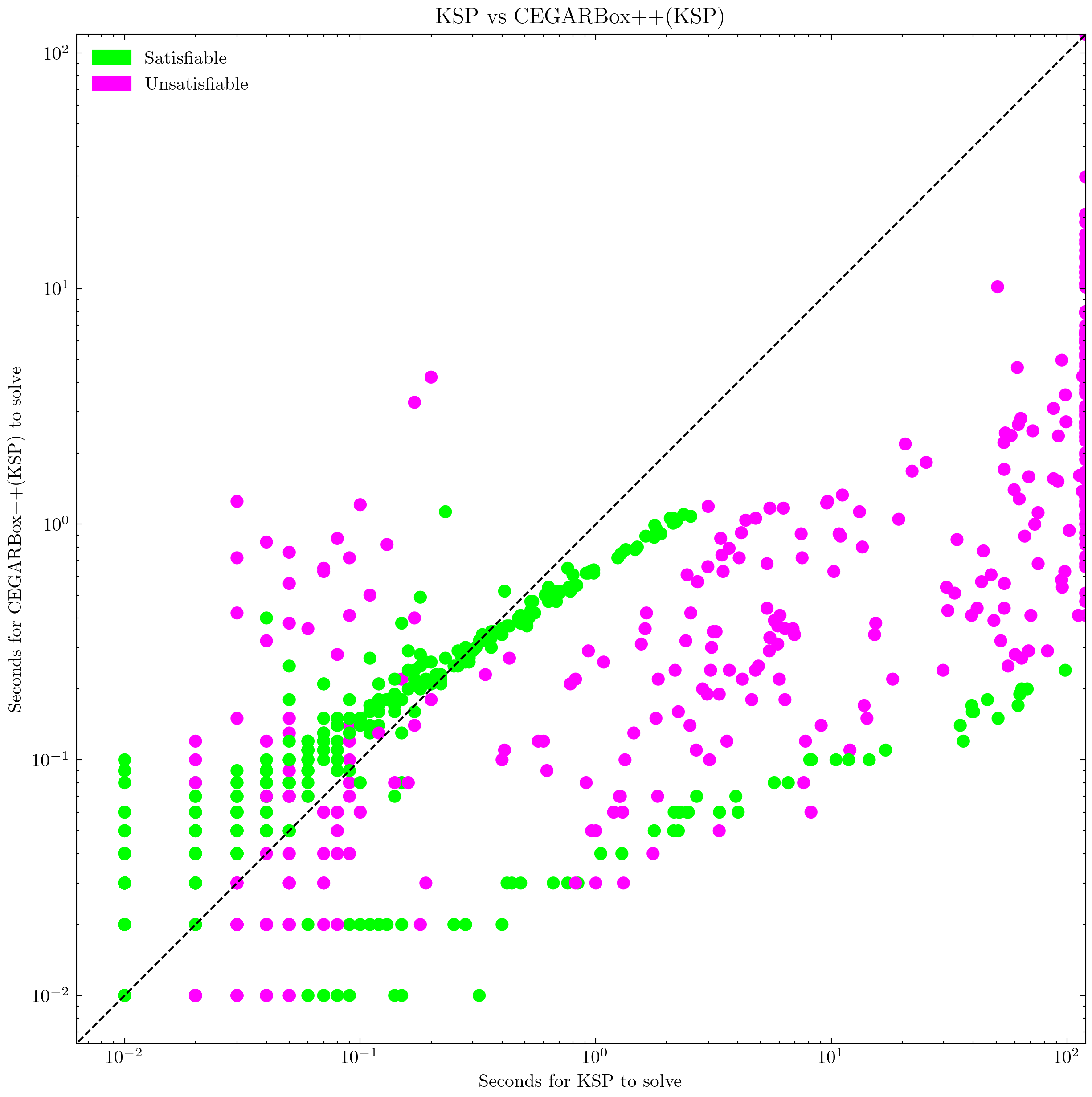}
  \caption{Performance of \cegarboxpp{(\KSP{})} and \KSP{} on the MQBF benchmarks. Problems on the
    top edge and the right edge indicate a timeout for \cegarboxpp{(\ksp)} and \KSP{} respectively.
    Problems in the bottom right triangle/top left triangle 
    are solved by \cegarboxpp{(\KSP{})} / \KSP{} quickest.}
    \label{fig:mqbf-after}
\end{figure}

Our experiments used an Intel i7-11700F@2.50 GHz CPU with 16GB of RAM. We used the same K-benchmarks as in
our original paper~\cite{DBLP:conf/tableaux/GoreK21}. Our code can be found here:
\url{https://github.com/cormackikkert/CEGARBoxCPP}.
We used \minisat{}~\cite{een2006minisat} as our SAT-solver, preliminary results showed that using other SAT-solvers
leads to negligible performance differences.

Using \cegarboxpp{} as our control, we tested our two variants that implement (\scsat) shortcuts: \cegarboxpp{(\ksp{})} and
\recarbox{}. Finally, as \cegarboxpp{(\ksp{})} runs an instance of \ksp{} ``under the hood'', we compare against \KSP{}
too. We use \ksp{(0.1.6)}, whereas our previous work~\cite{DBLP:conf/tableaux/GoreK21} used version 0.1.3 (the latest
at that time).

Since \cegarboxpp{(\ksp{})} uses multi-threading, CPU-time is not an appropriate measure of performance, so all reported times are
  ``wall time'', not CPU-time.

Figure~\ref{fig:all-after} shows that 
in all cases \recarbox{} is worse than \cegarboxpp{}, likely meaning that the extra work required to look for smaller models is wasted.
On the other hand we see that \cegarboxpp{(\ksp{})} outperforms \cegarboxpp{} by a large margin on the MQBF benchmarks and is equal on the 3CNF benchmarks. However, \cegarboxpp{} outperforms \cegarboxpp{(\ksp{}}) on the LWB-K benchmarks. This results in both provers being about equal overall. \recarbox{} performs better than \cegarboxpp{(\ksp)} on the LWB-K benchmarks, which we believe is due to the presence of massive formulae in these benchmarks: see our conclusion.

\paragraph{Synergy vs. Parallelism}

Finally, as the integration involves multithreading, one might wonder whether the performance gains are due to the parallelism itself (effectively ``cheating'' by doubling resources). Specifically for the MQBF benchmarks, is \cegarboxpp{} \textit{helping} to solve the satisfiable problems, or is it simply just \KSP{} solving them all (since it is best on satisfiable formulae)?

To this end, we compare \cegarboxpp{(\KSP{})} with \KSP{} in Figure~\ref{fig:mqbf-after}. We can see that on hard satisfiable benchmarks, \cegarboxpp{(\KSP{})} outperforms \KSP{}, meaning that \cegarboxpp{} is contributing. We can see no satisfiable timeouts for \cegarbox{(\KSP{})}, and even see a 100x improvement in some satisfiable problems, refuting the possibility that this improvement comes just from the doubling of resources due to parallelisation.

This highlights a distinction between a simple portfolio approach and our integration. While \KSP{} and \cegarboxpp{} generally
excel at different problem types (satisfiable vs. unsatisfiable), the integration appears to be better than the sum of their
parts. Communicating fixpoints has allowed us to achieve results that a non-communicating portfolio could not.

\section{Conclusions}

We can see that RECAR-tableaux doesn't perform that well. We believe this is because it operates on a much stricter form of
(\scsat) shortcuts, where it tries to create smaller models, as opposed to \cegartab{(\ksp{})}, which doesn't create
models, but just shows they exist. Creating smallest models is
intractable~\cite{DBLP:journals/ngc/ChenLP86,DBLP:journals/jar/Wolfram89} in general,
and
if they don't exist,
RECAR-tableaux simply wastes its
time.

Our approach, \cegartab{(\ksp{})}, for dealing with (\scsat) shortcuts shows a lot of promise, outperforming all solvers
on the MQBF benchmarks. We believe a better implementation would result in \cegartab{(\ksp{})} outperforming \cegartab{} across
the board. In particular, the issue with the current approach is that the \ksp{} normal forming procedure produces approximately
two times as many clauses as \cegartab{}, resulting in slowdowns on \textbf{massive} formulae, which are present in the LWB-K
benchmarks. These are why \recarbox{} performs better than \cegarboxpp{(\ksp)} on the LWB-K benchmarks.
\ksp{} can simplify this normal form (resulting in a smaller number of clauses than \cegartab{}), but on these massive
formulae the simplification times out. A more polished integration would use \cegartab{}'s normal forming procedure, and have
better and more real-time communication protocols. We imagine this would be capable of outperforming \cegartab{} on the 3CNF and
LWB-K benchmarks, whilst also full-solving the MQBF benchmarks. An implementation of such is easier said than done as \ksp{} has
certain requirements on its normal form and must produce extra clauses that are not needed for \cegartab{}. An alternative would
be to instead optimise the simplification process of \ksp{}, so it can produce reasonable normal forms quickly.

In general, this idea of searching for fixpoints is simply a method of detecting if $\check{\varphi}$ and $\varphi$ are equi-satisfiable, a normal idea in classical CEGAR, but more difficult in complicated CEGAR algorithms, where the under-approximation is a different type to the input (in our case a classical formula vs a modal formula). This fixpoint approach will be generally useful for other recursive CEGAR algorithms, such as CAQE~\cite{DBLP:conf/fmcad/RabeT15}, or PDR~\cite{DBLP:conf/vmcai/Bradley11}, though of course it is most effective in applications such as \cegartab{} where branching is involved, meaning any (\scsat) shortcut can possibly result in an exponential improvement in speed.

Whilst any approach could be used for fix-point detection, it is interesting to consider the implications of using resolution
here. Previously, SAT-based, tableaux-based and resolution-based solvers were the state-of-the-art solvers used for reasoning in
modal K. Now, \cegartab{(\ksp{})}, is an elegant combination of all these approaches, that is greater than the sum of its
parts. In particular, this highlights a key similarity of \cegartab{} and resolution methods. In \cegartab{} we learn clauses when
we need to (i.e.\ \textit{lazily}), whereas \ksp{} learns clauses iteratively (i.e.\ \textit{eagerly}).

We used
\ksp{} as
an oracle for \cegartab{}. From discussion with Cl{\'{a}}udia Nalon, Clare Dixon, Ullrich Hustadt, and Fabio Papacchini,
creating an efficient implementation by reversing the roles is not obvious.

\bibliographystyle{eptcs}
\bibliography{sn-bibliography}

\end{document}